\documentclass[letterpaper, 10 pt, conference]{ieeeconf}  
\IEEEoverridecommandlockouts
\usepackage{balance}
\usepackage{color}
\usepackage{caption}
\usepackage{enumerate}
\usepackage{cite}
\usepackage{empheq}
\usepackage{mathrsfs}

\usepackage{mathtools}

\usepackage{tikz}
\usepackage{circuitikz}

\usepackage[font=small,labelfont=bf]{caption}

\usepackage{subfig}


\makeatletter
\pgfcircdeclarebipole{}{\ctikzvalof{bipoles/interr/height 2}}{spst}{\ctikzvalof{bipoles/interr/height}}{\ctikzvalof{bipoles/interr/width}}{

    \pgfsetlinewidth{\pgfkeysvalueof{/tikz/circuitikz/bipoles/thickness}\pgfstartlinewidth}

    \pgfpathmoveto{\pgfpoint{\pgf@circ@res@left}{0pt}}
    \pgfpathlineto{\pgfpoint{.6\pgf@circ@res@right}{\pgf@circ@res@up}}
    \pgfusepath{draw}   
}

\def\pgf@circ@spst@path#1{\pgf@circ@bipole@path{spst}{#1}}
\tikzset{switch/.style = {\circuitikzbasekey, /tikz/to path=\pgf@circ@spst@path, l=#1}}
\tikzset{spst/.style = {switch = #1}}
\makeatother

\makeatletter
\let\proof\@undefined                        
\let\endproof\@undefined                  
\makeatother
\usepackage{graphicx,amssymb,amstext,amsmath,amsthm}

\usepackage[bookmarks=true]{hyperref}
\usepackage{algorithm,algorithmicx,algpseudocode}
\algnewcommand{\algorithmicgoto}{\textbf{go to}}%
\algnewcommand{\Goto}[1]{\algorithmicgoto~\ref{#1}}%
\algnewcommand{\LineComment}[1]{\Statex \(\triangleright\) #1}
\algnewcommand{\LineCommentN}[1]{\Statex \hspace{1cm}\(\triangleright\) #1}

\usepackage{multirow}
\usepackage{stfloats}

\newcommand{\argmin}{\operatornamewithlimits{arg\ min}}

\newtheorem{prop}{Proposition} 

\newtheorem{thm}{Theorem}
	\newtheorem{assumption}{Assumption}
\newtheorem{lem}{Lemma}
\newtheorem{defn}{Definition}

\newtheorem{problem}{Problem}

\usepackage{setspace}

\let\oldbibliography\thebibliography
\renewcommand{\thebibliography}[1]{%
  \oldbibliography{#1}%
}

\newcommand{\yong}[1]{{\color{black} #1}}
\newcommand{\moh}[1]{{\color{black} #1}}
\newcommand{\mo}[1]{{\color{black} #1}}
\newcommand{\fa}[1]{{\color{black} #1}}
\newcommand{\mk}[1]{{\color{black} #1}}
\newcommand{\mkr}[1]{{\color{black} #1}}
\newcommand{\md}[1]{{\color{black} #1}}
\newcommand{\yo}[1]{{\color{black} #1}}

\begin{document}

\title{\LARGE \bf \md{$\mathcal{H}_{\infty}$-Optimal} Interval Observer Synthesis for 
\yo{Uncertain} Nonlinear Dynamical Systems via Mixed-Monotone Decompositions} 

\author{%
Mohammad Khajenejad and Sze Zheng Yong\\
\thanks{
M. Khajenejad and S.Z. Yong are with the School for Engineering of Matter, Transport and Energy, Arizona State University, Tempe, AZ, USA. (e-mail: \{mkhajene, fshoaib, szyong\}@asu.edu.)}
\thanks{This work is partially supported by National Science Foundation grants CNS-1932066 and CNS-1943545.}
}

\maketitle
\thispagestyle{empty}
\pagestyle{empty}

\begin{abstract}
\yo{This paper introduces a novel $\mathcal{H}_{\infty}$-optimal interval observer synthesis for bounded-error/uncertain locally Lipschitz nonlinear continuous-time (CT) and discrete-time (DT) systems with noisy nonlinear  observations. Specifically, using mixed-monotone decompositions, the proposed observer is correct by construction, i.e., the interval estimates readily frame the true states without additional constraints or procedures. In addition, we provide sufficient conditions for input-to-state (ISS) stability of the proposed observer and for minimizing the $\mathcal{H}_{\infty}$ gain of the framer error system in the form of semi-definite programs (SDPs) with Linear Matrix Inequalities (LMIs) constraints. Finally, we compare the performance of the proposed $\mathcal{H}_{\infty}$-optimal interval observers with some  benchmark CT and DT interval observers.}     
 \end{abstract}

\section{Introduction} 
\md{\yo{Engineering applications, e.g., monitoring, system identification, control synthesis, and fault detection often} require knowledge of system states. However, due to the \yo{presence} of noise/\yo{uncertainties and/or inaccuracies} in sensor measurements, system states are \yo{usually} not \yo{exactly known}. 
This has motivated the design of state observers to estimate system states \yo{using} uncertain/noisy observations and system dynamics. In particular, for bounded-error settings, i.e., when  uncertainties are  \yo{set-valued (and distribution-free)}, \emph{interval observer} \yo{designs have recently gained much attention} due to their  \yo{simple principles and computational efficiency} \cite{wang2015interval}.}  
%

\md{
Recent years have produced an extensive body of seminal literature on the design of interval/set-valued observers for several classes of systems, e.g., linear, cooperative/monotone, Metzler and mixed-monotone dynamics \cite{wang2015interval,chebotarev2015interval,tahir2021synthesis,khajenejad2021intervalACC,khajenejad2020simultaneousCDC}.  
\yo{It has been noted that the design of interval observers that must simultaneously satisfy correctness (framer property) and stability/convergence is not a trivial task, even for linear systems \cite{chebotarev2015interval}. Thus, especially when the system dynamics is nonlinear, either relatively restrictive assumptions on system properties were required to guarantee the applicability 
of the proposed approaches, or monotone systems properties \cite{farina2000positive} need to be directly imposed to satisfy positivity/cooperative behavior of  
 the error dynamics.} 

This \yo{challenge} has been addressed for specific \yo{system} classes 
by leveraging M\"{u}ller's theorem \yo{or} interval arithmetic-based approaches \cite{kieffer2006guaranteed}, transformation to a positive system before designing an observer  (only for linear systems) \cite{cacace2014new} or applying time-invariant/varying state transformations \cite{tahir2021synthesis}.} 
\md{\yo{On the other hand}, the work in \cite{efimov2013interval} leveraged \emph{bounding functions} to design interval observers for a class of continuous-time nonlinear systems \yo{under some relatively restrictive} assumptions on the nonlinear dynamics, without providing a systematic approach to compute the bounding functions nor necessary/sufficient conditions for their existence. More recently, bounding/mixed-monotone decomposition functions were applied in \cite{tahir2021synthesis} to design interval \yo{observers} for nonlinear discrete-time dynamics, where conservative additive terms were added to the error dynamics to guarantee its positivity. Moreover, to best of our understanding, the resulting Linear Matrix Inequalities (LMIs) do not include the required conditions to guarantee that the computed bounding functions are decomposition functions. 

Decomposition functions were also applied in the authors' previous work \cite{khajenejad2021intervalACC,khajenejad2020simultaneousCDC} to design interval observers for nonlinear discrete-time systems \yo{under the (restrictive)} assumption of global Lipschitz continuity as well as additional sufficient structural system properties to guarantee stability. Further, the applied decomposition functions were not necessarily the tightest. In our \yo{preliminary} work \cite{moh2022intervalACC}, we proposed an interval observer design for noiseless nonlinear CT and DT systems based on tight remainder-form decomposition functions \cite{khajenejad2021tight}. Our goal is to extend the design in \cite{moh2022intervalACC} to \yo{noisy uncertain} 
nonlinear dynamics in this work.}   

\md {In particular, in this paper, we propose a unified framework to synthesize $\mathcal{H}_{\infty}$-\yo{optimal} and input-to-state stable (ISS) interval observes for a very broad range of locally Lipschitz bounded-error nonlinear CT and DT systems with \yo{noisy} nonlinear observation functions. Using \emph{remainder-form mixed-monotone decomposition functions} \cite{khajenejad2021tight}, we show that the states of the designed observer frame the true state trajectory of the system for all possible realizations of interval-valued noise/disturbance and initial states, i.e, the proposed observer is correct-by-construction, without imposing additional \yo{constraints or} assumptions. 
Further, \yo{we formulate semi-definite programs (SDPs) with LMI constraints for both CT and DT cases to minimize the $\mathcal{H}_{\infty}$-gain of the framer error (interval width) system and to ensure input-to-state stability of the correct-by-construction framers, which can be computed offline to find the stabilizing observer gains.}} 
 \section{Preliminaries}
 
 {\emph{{Notation}.}} $\mathbb{R}^n,\mathbb{R}^{n  \times p},\mathbb{D}_n,\mathbb{N},\mathbb{N}_n$ denote the $n$-dimensional Euclidean space and the sets of $n$ by $p$ matrices, $n$ by $n$ diagonal matrices, natural numbers and natural numbers up to $n$, respectively, while $\mathbb{M}_n$ 
 denote\yong{s} the set of all $n$ by $n$ Metzler\footnote{A Metzler matrix is a square matrix in which all the off-diagonal components are nonnegative (equal to or greater than zero).} \mo{matrices}. 
 For $M \in \mathbb{R}^{n \times p}$, $M_{ij}$ denotes $M$'s entry in the $i$'th row and the $j$'th column, $M^+\triangleq \max(M,\mathbf{0}_{n,p})$, $M^-=M^+-M$ and $|M|\triangleq M^++M^-$, where $\mathbf{0}_{n,p}$ is the zero matrix in $\mathbb{R}^{n \times p}$, \yong{while $\textstyle{\mathrm{sgn}}(M) \in \mathbb{R}^{n \times p}$ is the element-wise sign of $M$ with $\textstyle{\mathrm{sgn}}(M_{ij})=1$ if $M_{ij} \geq 0$ and $\textstyle{\mathrm{sgn}}(M_{ij})=-1$, otherwise.} 
 Further, if $p=n$, $M^\text{d}$ denotes a diagonal matrix whose diagonal coincides with the diagonal of $M$, $M^\text{nd} \triangleq M-M^\text{d}$ and $M^{\text{m}} \triangleq M^\text{d}+|M^\text{nd}|$, \mk{while $M \succ 0$ and $M \prec 0$ (or $M \succeq 0$ and $M \preceq 0$) denote that $M$ is positive and negative (semi-)definite, respectively}. \md{Finally, a function $f:S \subseteq \mathbb{R}^n \to \mathbb{R}$, where $0 \in S$, is positive definite if $f(x) >0$ for all $x \in S\yo{\, \setminus \{0\}}$, and $f(0)=0$.}

Next, we introduce some useful definitions and results.
\begin{defn}[Interval, Maximal and Minimal Elements, Interval Width]\label{defn:interval}
{An (multi-dimensional) interval {$\mathcal{I} \triangleq [\underline{\md{z}},\overline{\md{z}}]  \subset 
\mathbb{R}^n$} is the set of all real vectors $\md{z \in \mathbb{R}^{n_z}}$ that satisfies $\underline{\md{z}} \le \md{z} \le \overline{\md{z}}$, where $\underline{\md{z}}$, $\overline{\md{z}}$ and $\|\overline{\md{z}}-\underline{\md{z}}\|\mk{_{\infty}\triangleq \max_{i \in \{1,\cdots,\md{n_z}\}}\md{|z_i|}}$ are called minimal vector, maximal vector and \mk{interval} width of $\mathcal{I}$, respectively}. \yong{An interval matrix can be defined similarly.} 
\end{defn}
\begin{prop}\cite[Lemma 1]{efimov2013interval}\label{prop:bounding}
Let $A \in \mathbb{R}^{n \times p}$ and $\underline{x} \leq x \leq \overline{x} \in \mathbb{R}^n$. Then\moh{,} $A^+\underline{x}-A^{-}\overline{x} \leq Ax \leq A^+\overline{x}-A^{-}\underline{x}$. As a corollary, if $A$ is non-negative, $A\underline{x} \leq Ax \leq A\overline{x}$. 
\end{prop}

\begin{defn}[Jacobian Sign-Stability] \label{defn:JSS}
A mapping $f :\md{\mathcal{Z}} \subset \mathbb{R}^\md{n_z} \to  \mathbb{R}^{p}$ is (generalized) Jacobian sign-stable (JSS), if its (generalized) Jacobian matrix entries \fa{do} not change signs on its domain, i.e., if either of the following hold: 
\begin{align*}
&\forall \md{z \in \mathcal{Z}}, \forall i \in \mathbb{N}_p,\forall j \in \mathbb{N}_{n_z} , J^f_{ij}(\md{z}) \geq 0 \ \text{(positive JSS)},\\  
&\forall \md{z \in \mathcal{Z}}, \forall i \in \mathbb{N}_p,\forall j \in \mathbb{N}_{n_z} , J^f_{ij}(\md{z}) \leq 0 \  \text{(negative JSS)},
\end{align*} 
where $J^f(\md{z})$ denotes the Jacobian matrix of $f$ at $\md{z \in \mathcal{Z}}$. 
\end{defn}
\begin{prop}[Jacobian Sign-Stable Decomposition, \md{\cite[Proposition 2]{moh2022intervalACC}}]\label{prop:JSS_decomp}
Let $f :\md{\mathcal{Z}} \subset \mathbb{R}^{n_z} \to  \mathbb{R}^{p}$ and suppose $\forall \md{z \in \mathcal{Z}}, J^f(\md{z}) \in [\underline{J}^f,\overline{J}^f]$, where $\underline{J}^f,\overline{J}^f$ are known matrices in $\mathbb{R}^{p \times n_z}$. Then, $f$ can be decomposed into a (remainder) affine mapping $H\md{z}$ and a JSS mapping $\mu (\cdot)$, in an additive form: 
\begin{align}\label{eq:JSS_decomp}
\forall \md{z \in \mathcal{Z}},f(\md{z})=\mu(\md{z})+H\md{z},
\end{align}
 where $H$ is a matrix in $\mathbb{R}^{p \times n_z}$, that satisfies the following 
 \begin{align}\label{eq:H_decomp}
 \forall (i,j) \in \mathbb{N}_p \times \mathbb{N}_{n_z}, H_{ij}=\overline{J}^f_{i,j} \ \lor  H_{ij}=\underline{J}^f_{ij}.    
 \end{align}
\end{prop}
\begin{defn}[Mixed-Monotonicity {and} Decomposition Functions] \cite[Definition 1]{abate2020tight},\cite[Definition 4]{yang2019sufficient} \label{defn:dec_func}
Consider the \md{bounded-error} dynamical system {with initial state $x_0 \in \mathcal{X}_0 \triangleq [\underline{x}_0,\overline{x}_0] \md{\subset \mathbb{R}^{n}}$} \md{and process noise $w_t \in \mathcal{W} \triangleq [\underline{w},\overline{w}] \subset \mathbb{R}^{n_w} $}:
\begin{align}\label{eq:mix_mon_def}
x_t^+= \md{g(z_t) \triangleq} g(x_t,\md{w_t}), \md{z_t \triangleq [x^\top_t \ w_t]^\top}
\end{align}
where $x_t^+ \triangleq x_{t+1}$ if \eqref{eq:mix_mon_def} is a DT {system} and $x_t^+ \triangleq \dot{x}_t$ if \eqref{eq:mix_mon_def} is a CT system. Moreover, ${g}:\md{\mathcal{Z}} \subset \mathbb{R}^\md{n_z} \to \mathbb{R}^{\md{n}}$ is the vector field 
with \md{augmented} state $\md{z_t} \in \md{\mathcal{Z} \triangleq \mathcal{X} \times \mathcal{W} \subset \mathbb{R}^{n_z}}$ \yong{as its domain}, \md{where $\mathcal{X}$ is the entire state space and  $n_z=n+n_w$}. 

Suppose \eqref{eq:mix_mon_def} is a DT system. 
Then, a mapping $g_d:\md{\mathcal{Z}\times \mathcal{Z}} \to \mathbb{R}^{p}$ is 
a {DT mixed-monotone} decomposition function with respect to $g$, 
if i) $g_d(\md{z,z})=g(\md{z})$, ii) $g_d$ is monotone increasing in its first 
argument, i.e., $\hat{\md{z}}\ge \md{z} \Rightarrow g_d(\hat{\md{z}},\md{z}') \geq g_d(\md{z,z}')$, and iii) {$g_d$ is monotone decreasing in its second argument, i.e., $\hat{\md{z}}\ge \md{z} \Rightarrow g_d(\md{z}',\hat{\md{z}}) \leq g_d(\md{z}',\md{z}).$}

\yong{On the other hand,} if \eqref{eq:mix_mon_def} is a CT system, 
a mapping $g_d:\md{\mathcal{Z}\times \mathcal{Z}} \to \mathbb{R}^{p}$ is a {CT mixed-monotone} decomposition function with respect to $g$, 
if i) $g_d(\md{z},\md{z})=g(\md{z})$, ii) $g_d$ is monotone increasing in its first 
argument with respect to ``off-diagonal'' arguments, i.e., $\forall (i,j) \in \mathbb{N}_{\md{n}} \times \mathbb{N}_{\md{n_z}} \land i \ne j,\md{\hat{z}}_j\ge \md{z}_j, \md{\hat{z}_i= z_i}  \Rightarrow g_{d,i}(\hat{\md{z}},\md{z}') \geq g_{d,i}(\md{z,z}')$, and iii) {$g_d$ is monotone decreasing in its second argument, i.e., $\hat{\md{z}}\ge \md{z} \Rightarrow g_d(\md{z}',\hat{\md{z}}) \leq g_d(\md{z}',\md{z}).$}   
\end{defn}
\begin{defn}[\md{Bounded-Error} Embedding {Systems}]\label{def:embedding}
For an $n$-dimensional 
system \eqref{eq:mix_mon_def} {with any decomposition functions ${g}_d(\cdot)$, its \emph{embedding system is the following} 
$2n$-dimensional system {with initial condition $\begin{bmatrix} \overline{x}_0^\top & \underline{x}_0^\top\end{bmatrix}^\top$}:}
\begin{align} \label{eq:embedding}
\md{\begin{bmatrix}{\underline{x}}_t^+ \\ {\overline{x}}_t^+ \end{bmatrix}=\begin{bmatrix}  \underline{g}_d(\begin{bmatrix}(\underline{x}_t)^\top \, \underline{w}^\top \end{bmatrix}^\top\hspace{-0.05cm},\begin{bmatrix}(\overline{x}_t)^\top \, \overline{w}^\top\end{bmatrix}^\top) \\  \overline{g}_d(\begin{bmatrix}(\overline{x}_t)^\top \, \overline{w}^\top\end{bmatrix}^\top\hspace{-0.05cm},\begin{bmatrix}(\underline{x}_t)^\top \, \underline{w}^\top \end{bmatrix}^\top) \end{bmatrix}.} 
\end{align}
\end{defn}

\begin{prop}[State Framer Property]\label{cor:embedding} \cite[Proposition 3]{khajenejad2021tight}
{Let system \eqref{eq:mix_mon_def} with initial state $x_0 \in \mathcal{X}_0 \triangleq  [\underline{x}_0,\overline{x}_0]$ \md{and $w_t \in \mathcal{W}$} be mixed-monotone with an 
 embedding system \eqref{eq:embedding} with respect to 
${g}_d$. Then, for all $t \in \yong{\mathbb{T}}$, $R^g(t,\mathcal{X}_0,\md{\mathcal{W}}) \subset \mathcal{X}_t \triangleq [\underline{x}_t,\overline{x}_t]$, 
where $R^g(t,\mathcal{X}_0,\md{\mathcal{W}}) \triangleq
\{\mu_g(t, x_0,\md{{w}}) \mid x_0 \in \mathcal{X}_0,  \forall t \in \yong{\mathbb{T}},\md{\forall w \in \mathcal{W}}\}$ is the reachable set at time $t$ of \eqref{eq:mix_mon_def} 
when initialized within $\mathcal{X}_0$, $\mu_g(t,x_0,\md{{w}})$ is the true state trajectory function of system \eqref{eq:mix_mon_def} and $(\overline{x}_t,\underline{x}_t)$ is the solution to the embedding system \eqref{eq:embedding}, \yong{with $\mathbb{T} \in \mathbb{R}_{\ge 0}$ for CT systems and $\mathbb{T} \in \{0\} \cup \mathbb{N}$ for DT systems}. Consequently,
the system state trajectory $x_t=\mu_g(t,x_0,\md{{w}})$ 
satisfies $\underline{x}_t \le x_t \le \overline{x}_t, {\forall t \geq 0}, \md{\forall w \in \mathcal{W}}$, i.e., \yong{it} is \emph{framed} by $\mathcal{X}_t \triangleq [\underline{x}_t,\overline{x}_t]$.}
\end{prop}

\begin{prop}[Tight and Tractable Decomposition Functions for JSS Mappings, \md{\cite[Proposition 4]{moh2022intervalACC}}]\label{prop:tight_decomp}
Let $\mu:\md{\mathcal{Z}} \subset \mathbb{R}^{n_z} \to \mathbb{R}^p$ be a JSS mapping on its domain. Then, it admits a tight decomposition function that has the following form: 
\begin{align}\label{eq:JJ_decomp}
\forall i \in \mathbb{N}_p, \mu_{d,i}(\md{z}_1,\md{z}_2)\hspace{-.1cm}=\hspace{-.1cm}\mu_i(D^i\md{z}_1\hspace{-.1cm}+\hspace{-.1cm}(I_{n_z}\hspace{-.1cm}-\hspace{-.1cm}D^i)\md{z}_2), 
\end{align}
\md{for any {ordered} $\md{z_1, z_2 \in \mathcal{Z}}$}, where $D^i \in \mathbb{D}_\md{n_z}$ is a binary diagonal matrix determined by which vertex of the interval \md{$[{z}_2,{z}_1]$ \yo{that maximizes} \md{or $[z_1,z_2]$} \yo{that minimizes}} 
the JSS function $\mu_i(\cdot)$ and \md{\yo{can be} computed} as follows:
\begin{align}\label{eq:Dj}
D^i=\textstyle{\mathrm{diag}}(\max(\textstyle{\mathrm{sgn}}(\yong{\overline{J}^{\mu}_i}),\mathbf{0}_{1,\md{n_z}})).
\end{align}
\end{prop}
\section{Problem Formulation} \label{sec:Problem}
\noindent\textbf{\emph{System Assumptions.}} 
Consider the following \yo{uncertain} nonlinear  continuous-time (CT) or discrete-time (DT) system:  
\begin{align} \label{eq:system}
\begin{array}{ll}
\mathcal{G}: \begin{cases} {x}_t^+ = \hat{f}(x_t,\md{w_t},u_t) \triangleq f(x_t,\md{w}_t),   \\
                                              y_t = \hat{h}(x_t,\md{v_t},u_t) \triangleq h(x_t,\md{v_t}), 
                                              \end{cases} \hspace{-.2cm} x_t\in \mathcal{X}, t \in \yong{\mathbb{T}},
\end{array}\hspace{-0.2cm}
\end{align}
where $x_t^+=\dot{x}_t, \yong{\mathbb{T}} = \mathbb{R}_{\ge 0}$ if $\mathcal{G}$ is a CT and $x_t^+=x_{t+1}, \yong{\mathbb{T}}= \{0\}\cup \mathbb{N}$, if $\mathcal{G}$ is a DT system. Moreover, $x_t \in \mathcal{X} \subset \mathbb{R}^n$, \md{$w_t \in \mathcal{W} \triangleq [\underline{w},\overline{w}] \subset \mathbb{R}^{n_w},v_t \in \mathcal{V} \triangleq [\underline{v},\overline{v}] \subset \mathbb{R}^{n_v}$}, $u_t \in \mathbb{R}^s$ and $y_t \in \mathbb{R}^l$ are continuous state, \md{process noise, measurement disturbance,} known (control) input and output (measurement) signals. Furthermore, $\hat{f}:\mathbb{R}^n \md{\times \mathbb{R}^{n_w}} \times \mathbb{R}^s \to \mathbb{R}^n$ and $\hat{h}:\mathbb{R}^n \md{\times \mathbb{R}^{n_v}} \times \mathbb{R}^s \to \mathbb{R}^l$ are nonlinear state vector field and observation/constraint \yong{functions/}mappings, respectively, from which, 
the \yong{functions/}mappings $f:\mathbb{R}^n \md{\times \mathbb{R}^{n_w}} \to \mathbb{R}^n$ and $g:\mathbb{R}^n \md{\times \mathbb{R}^{n_v}} \to \mathbb{R}^l$ are well-defined \yong{since the input signal $u_t$ is known}.  
We are interested in estimating the trajectories
of the plant $\mathcal{G}$ in \eqref{eq:system}, when they are initialized in an interval
$\mathcal{X}_0 \md{\triangleq [\underline{x}_0,\overline{x}_0]} \subset \mathcal{X} \subset \mathbb{R}^n$.
We 
also assume the following: 
\begin{assumption} \label{ass:initial_interval}
 The initial state $x_0$ satisfies $x_0 \in \mathcal{X}_0 = [ \underline{x}_0,\overline{x}_0]$, where $\underline{x}_0$ and $\overline{x}_0$ {are} known initial state bounds. 
 \end{assumption}
 \begin{assumption}\label{ass:mixed_monotonicity}
 The 
 \yong{mappings} $f(\cdot)$ and $h(\cdot)$ are known, \mk{differentiable}, locally Lipschitz\mk{\footnote{\mk{Both assumptions of locally Lipschitz continuity and differentiability are mainly made for ease of exposition and can be relaxed to a much weaker continuity assumption (cf. \cite{khajenejad2021tight} for more details).}}} and mixed-monotone in their domain with \emph{a priori} known upper and lower bounds for their Jacobian matrices, $\overline{J}^{f},\underline{J}^{f} \in \mathbb{R}^{n \times \md{n_z}}$ and $\overline{J}^{h},\underline{J}^{h} \in \mathbb{R}^{l \times \md{n_{\zeta}}}$, respectively, \md{where $n_z=n+n_w$ and $n_{\zeta}=n+n_v$}.
\end{assumption}
\begin{assumption}\label{ass:known_input_output}
\md{$\mathcal{X}_0, \mathcal{W},\mathcal{V} $ and} the values of the input $u_t$ and output/measurement $y_t$ signals are known\yo{/given} at all times. 
\end{assumption}
Further, we formally define the notions of \emph{framers}, \emph{correctness} and \emph{stability} that are used throughout the paper. 
\begin{defn}[Correct Interval \mk{Framers}]\label{defn:framers}
Suppose Assumptions \ref{ass:mixed_monotonicity} \md{and} \ref{ass:known_input_output} hold. Given the nonlinear plant \eqref{eq:system}, 
the mappings/signals $\overline{x},\underline{x}: \mo{\mathbb{T}} \to \mathbb{R}^n$ are called upper and lower framers for the states of System \eqref{eq:system}, if 
\begin{align}\label{eq:correctness}
\forall t \in \yong{\mathbb{T}}, \md{\forall w_t \in \mathcal{W},\forall v_t \in \mathcal{V}}, \ \underline{x}_t \leq x_t \leq \overline{x}_t.
\end{align}
In other words, starting from the initial interval $\underline{x}_0 \leq x_0 \leq \overline{x}_0$, the true state of the system in \eqref{eq:system}, $x_t$, is guaranteed to evolve within the interval flow-pipe $[\underline{x}_t,\overline{x}_t]$, for all $(t,\md{w_t,v_t)} \in \md{\mathbb{T} \times \mathcal{W} \times \mathcal{V}}$. Finally, any dynamical system whose states are correct framers for the states of the plant $\mathcal{G}$, i.e., any (tractable) algorithm that returns upper and lower framers for the states of plant $\mathcal{G}$ is called a \emph{correct} interval \mk{framer} for system \eqref{eq:system}. 
\end{defn}
\begin{defn}[\mk{Framer} Error]\label{defn:error}
Given 
state framers \mk{$\underline{x}_t \leq \overline{x}_t$}, $\varepsilon : \mo{\mathbb{T}} \to \mathbb{R}^n$, \md{denoting} the interval width \mk{of $[\underline{x}_t,\overline{x}_t]$ (cf. Definition \ref{defn:interval})},  
 is called the \mk{framer} error. It can be easily verified that 
correctness (cf. Definition \ref{defn:framers}) implies that $\varepsilon_t \geq 0, \forall t \in \mo{\mathbb{T}}.$  
\end{defn}
\begin{defn}[\md{Input-to-State} Stability and \mk{Interval Observer}]\label{defn:stability}
An interval \mk{framer} is \md{input-to-state} stable \md{(ISS)}, if the \mk{framer} error (cf. Definition \ref{defn:error}) \md{is bounded as follows: 
\begin{align}
\forall t \in \mathbb{T}, \|\varepsilon_t\|_{2} \leq \beta(\|\varepsilon_0\|_{2},t)+\rho(\|\Delta\|_{\ell_\infty}),
\end{align} 
where $\Delta \triangleq [\Delta w^\top \ \Delta v^\top]^\top \triangleq [(\overline{w}-\underline{w})^\top \ (\overline{v}-\underline{v})^\top]^\top$,
$\beta$ and $\rho$ are functions of classes\footnote{\md{A function $\alpha: \mathbb{R}_+ \to \mathbb{R}_+$ is of class $\mathcal{K}$ if it is continuous, positive definite, and strictly increasing and is of class $\mathcal{K}_{\infty}$ if it is also unbounded. Moreover, $\lambda : \mathbb{R}_+ \to \mathbb{R}_+$ is of class $\mathcal{KL}$ if for each fixed $t\geq 0$, $\lambda(\cdot,t)$ is of class $\mathcal{K}$ and for each fixed $s \geq 0$, $\lambda(s,t)$ decreases to zero as $t \to \infty.$}} $\mathcal{KL}$ and $\mathcal{K}_{\infty}$, respectively, and $\|\Delta\|_{\ell_{\infty}} \triangleq \sup_{t \in [0,\infty]}\|\Delta_t\|_2=\|\Delta\|_2$ is the $\ell_{\infty}$ signal norm.}
\md{An ISS interval framer is called an interval observer.}
\end{defn}
\md{
\begin{defn}[$\mathcal{H}_{\infty}$-Optimal Interval Observer Synthesis]\label{defn:H_inf}
An interval framer design $\hat{\mathcal{G}}$ is $\mathcal{H}_{\infty}$-optimal 
if the $L_{\infty}$ gain of the framer error system $\tilde{\mathcal{G}}$, i.e., $\|\tilde{\mathcal{G}}\|_{L_{\infty}}$ is minimized, where 
\begin{align}\label{eq:H_inf_Def}
\|\tilde{\mathcal{G}}\|_{L_{\infty}} \triangleq \sup {\{}\frac{\|\varepsilon\|_{\ell_2}}{\|\Delta\|_{\ell_2}},\Delta \ne 0{\}}, 
\end{align}
and $\|s\|_{\ell_2} \triangleq \sqrt{\int_{0}^{\infty}\|s_t\|_2^2dt}$ is the $\ell_2$ signal norm for $s \in \{\varepsilon,\Delta\}$.
\end{defn}}
The 
observer design problem 
 can be stated as follows:
\begin{problem}\label{prob:SISIO}
Given the nonlinear system in \eqref{eq:system}, as well as Assumptions \md{\ref{ass:mixed_monotonicity} and \ref{ass:known_input_output}}, 
synthesize \md{an ISS and \md{$\mathcal{H}_{\infty}$-optimal} interval observer (cf. Definitions \ref{defn:framers}--\ref{defn:H_inf})}. 
\end{problem}

\section{Proposed Interval Observer} \label{sec:observer}
\subsection{Interval Observer Design} \label{sec:obsv}
Given the nonlinear plant $\mathcal{G}$, in order to address Problem \ref{prob:SISIO}, we propose an interval observer (cf. Definition \ref{defn:framers}) for $\mathcal{G}$ through the following dynamical system \md{$\hat{\mathcal{G}}$}:
\md{
\begin{align}\label{eq:observer}
\begin{array}{rlll}
\underline{x}_t^+&\hspace{-.2cm}=\hspace{-.1cm}(A-LC)^\uparrow \underline{x}_t\hspace{-.1cm}-\hspace{-.1cm}(A\hspace{-.1cm}-\hspace{-.1cm}LC)^\downarrow \overline{x}_t
\hspace{-.1cm}+\hspace{-.1cm}\phi_d(\underline{\md{x}}_t,\underline{w},\overline{\md{x}}_t,\overline{w}) \\
&\hspace{-.2cm}-L^+\psi_d(\overline{\md{x}}_t,\overline{v},\underline{\md{x}}_t,\underline{v})+L^-\psi_d(\underline{\md{x}}_t,\underline{v},\overline{\md{x}}_t,\overline{v})+Ly_t \\
&\hspace{-.2cm}+B^+\underline{w}-B^-\overline{w}+(LD)^-\underline{v}-(LD)^+\overline{v},\\
\overline{x}_t^+&\hspace{-.2cm}=\hspace{-.1cm}(A\hspace{-.1cm}-\hspace{-.1cm}LC)^\uparrow \overline{x}_t\hspace{-.1cm}-\hspace{-.1cm}(A-LC)^\downarrow \underline{x}_t
\hspace{-.1cm}+\hspace{-.1cm}\phi_d(\overline{\md{x}}_t,\overline{w},\underline{\md{x}}_t,\underline{w}) \\
&\hspace{-.2cm}-L^+\psi_d(\underline{\md{x}}_t,\underline{v},\overline{\md{x}}_t,\overline{v})+L^-\psi_d(\overline{\md{x}}_t,\overline{v},\underline{\md{x}}_t,\underline{v})+Ly_t \\
&\hspace{-.2cm}+B^+\overline{w}-B^-\underline{w}+(LD)^-\overline{v}-(LD)^+\underline{v}, 
 \end{array}
\end{align}
}\vspace{-0.2cm}

\noindent where if $\mathcal{G}$ is a CT system, then
\begin{align}\label{eq:T_CT}
\begin{array}{rl}
\overline{x}_t^+\hspace{-.2cm}&\triangleq \dot{\overline{x}_t},(A-LC)^\uparrow \triangleq (A-LC)^\text{d}+(A-LC)^{\text{nd}+},\\
\underline{x}_t^+\hspace{-.2cm}&\triangleq \dot{\underline{x}_t},(A-LC)^\downarrow \triangleq (A-LC)^{\text{nd}-},
\end{array}
\end{align}
and if $\mathcal{G}$ is a DT system, then
\begin{align}\label{eq:T_DT}
\begin{array}{rl}
\overline{x}_t^+&\triangleq \overline{x}_{t+1},(A-LC)^\uparrow \triangleq (A-LC)^+,\\
\underline{x}_t^+&\triangleq \underline{x}_{t+1},(A-LC)^\downarrow \triangleq (A-LC)^{-}.
\end{array}
\end{align}
Moreover, 
$A \in \mathbb{R}^{n \times n}$, \md{$B \in \mathbb{R}^{n \times n_w}$}, $C \in \mathbb{R}^{l \times n}$ \md{and $D \in \mathbb{R}^ {l \times n_v}$} are chosen such that the following decompositions hold (cf. Definition \ref{defn:JSS} and Proposition \ref{prop:JSS_decomp}): \md{$\forall x,w,v \in \mathcal{X} \times \mathcal{W} \times \mathcal{V},$}
\vspace{-.2cm}
\md{
\begin{align} \label{eq:JSS_decom}
\begin{cases}
f(x,w)=Ax+Bw+\phi(x,w), \\ h(x,v)=Cx+Dv+\psi(x,v), \end{cases} \hspace{-.4cm} s.t. \ \phi,\psi \ \text{are JSS}. 
\end{align} 
} \vspace{-0.2cm}

\noindent Furthermore, $\phi_d:\mathbb{R}^\md{2n_z} \to \mathbb{R}^n$ and $\psi_d:\mathbb{R}^\md{2n_{\zeta}}  \to \mathbb{R}^l$ are tight mixed-monotone decomposition functions of $\phi$ and $\psi$, respectively (cf. Definition \ref{defn:dec_func} and Propositions \ref{cor:embedding}--\ref{prop:tight_decomp}). Finally, $L \in \mathbb{R}^{n \times l}$ is the observer gain matrix, designed via Theorem \ref{thm:stability}, such that the proposed observer $\hat{\mathcal{G}}$ possesses the desired properties discussed in the following subsections.
\subsection{Observer Correctness (Framer Property)}\vspace{-0.05cm}
Our strategy is to design a \emph{correct by construction} interval observer for plant $\mathcal{G}$. To accomplish this goal, first, note that from \eqref{eq:observer} and \eqref{eq:JSS_decom}
we have $y_t-Cx_t\md{-Dv_t}-\psi(x_t,\md{v_t})=0$, and so $L(y_t-Cx_t\md{-Dv_t}-\psi(x_t,\md{v_t}))=0$, for any $L \in \mathbb{R}^{n \times l}$. Adding this ``zero" term to the right hand side of \eqref{eq:system} \yong{and applying} 
\eqref{eq:JSS_decom}
yield the following equivalent system to $\mathcal{G}$:  
\begin{align} \label{eq:eqiv_sys}
\hspace{-.25cm}x_t^+\hspace{-.15cm}=\hspace{-.1cm}(\hspace{-.05cm}A\hspace{-.1cm}-\hspace{-.1cm}LC\hspace{-.05cm})x_t\hspace{-.1cm}+\hspace{-.1cm}\md{Bw_t\hspace{-.1cm}-\hspace{-.1cm}LDv_t}\hspace{-.1cm}+\hspace{-.1cm}\phi(\hspace{-.05cm}x_t,\hspace{-.05cm}\md{w_t}\hspace{-.05cm})\hspace{-.1cm}-\hspace{-.15cm}L\psi(x_t,\md{v_t}\hspace{-.05cm})\hspace{-.1cm}+\hspace{-.15cm}Ly_t
\end{align}
From now on, we are interested in computing embedding systems, in the sense of Definition \ref{def:embedding}, for the system in \eqref{eq:eqiv_sys}, so that by Proposition \ref{cor:embedding}, the state trajectories of \eqref{eq:eqiv_sys} are ``framed" by the state trajectories of the computed embedding system. To do so, we split the right hand side of \eqref{eq:eqiv_sys} (except for $L\md{y_t}$ that is \yo{a known signal}) 
into two constituent systems: the \md{affine} constituent $(A-LC)x_t\md{+Bw_t-LDv_t}$, and the nonlinear constituent, $\phi(x_t,\md{w_t})-L\psi(x_t,\md{v_t})$. Then, we compute embedding systems for each constituent, separately. Finally, we add the computed embedding systems to construct an embedding system for \eqref{eq:eqiv_sys}. We start with framing the \md{affine} constituent through the following lemma.
\begin{lem}[\md{Affine} Embedding]\label{lem:linear_bounding}
Consider a dynamical system ${\mathcal{G}}_{\ell}$ in the form of \eqref{eq:mix_mon_def}, with 
state equation ${f}_{\ell}(\md{\xi_t})=(A-LC)x_t\md{+Bw_t-LDv_t}$ \md{(with $\xi \triangleq [x^\top \ w^\top \ v^\top]^\top$)}. Then, a tight decomposition function (cf. Definition \ref{defn:dec_func}) for ${\mathcal{G}}_{\ell}$ \md{and any ordered pair of $\xi_1,\xi_2$} can be computed as follows:
\md{
\begin{align} \label{eq:Lin_dec}
\begin{array}{rl}
&{f}_{\ell d}(\md{\xi_1,\xi_2})=(A-LC)^\uparrow x_1-(A-LC)^\downarrow x_2 \\
                                                  &+B^+w_1-B^-w_2+(LD)^-v_1-(LD)^+v_2,
\end{array}
\end{align}
}\vspace{-0.2cm}

\noindent where $(A-LC)^\uparrow $ and $(A-LC)^\downarrow$ are given in \eqref{eq:T_CT} and \eqref{eq:T_DT} for CT and DT systems, respectively. 
\end{lem}
\begin{proof} \md{The proof is similar to the proof of \cite[Lemma 1]{moh2022intervalACC}, with the slight modification of considering the extra noise terms $B^+w_1-B^-w_2$ that is non-decreasing in $w_1$ and non-\yo{in}creasing in $w_2$ due to the non-negativity of $B^+$ and $B^-$, as well as $(LD)^-v_1-(LD)^+v_2$ that is non-decreasing in $v_1$ and non-\yo{in}creasing in $v_2$ due to the non-negativity of $(LD)^+$ and $(LD)^-$.}
\end{proof}   
Next, we compute an embedding system for the nonlinear constituent system 
$\phi(x_t,\md{w_t})-L\psi(x_t,\md{v_t})$, as follows.
\begin{lem}[Nonlinear Embedding]\label{lem:nonlinear_bounding}
Consider a dynamical system ${\mathcal{G}}_{{\nu}}$ in the form of \eqref{eq:mix_mon_def}, with domain $\mathcal{X}$ and state equation ${f}_{\nu}(x_t,\md{w_t},\md{v_t})=\phi(x_t,\md{w_t})-L\psi(x_t,\md{v_t})$. Then, a decomposition function (cf. Definition \ref{defn:dec_func}) for ${\mathcal{G}}_{\nu}$ can be computed as follows \md{(with $\xi \triangleq [x^\top \ w^\top \ v^\top]^\top$)}:
\md{
\begin{align}\label{eq:Lin_dec}
\begin{array}{rl}
&{f}_{\nu d}(x_1,w_1,v_1,x_2,w_2,v_2)=\phi_d({x}_1,w_1,{x}_2,w_2)\\
&-L^+\psi_d({x}_2,v_2,{x}_1,v_1)+L^-\psi_d({x}_1,v_1,{x}_2,v_2),
\end{array}
\end{align} 
}\vspace{-0.2cm}

\noindent where $\phi_d,\psi_d$ are tight decomposition functions for the JSS mapping $\phi,\psi$, computed via Proposition \ref{prop:tight_decomp}. 
\end{lem}
\begin{proof}
${f}_{\nu d}$ is increasing in \md{$\xi_1$} since it is a summation of three increasing mappings in \md{$\xi_1$}, including $\phi_d(x_1,\md{w_1},x_2,\md{w_2})$ (a decomposition function that by construction is increasing in \md{$(x_1,w_1)$} \md{and hence in $\xi_1=(\md{x_1,w_1,v_1})$}), $-L^+\psi_d(\md{{x}_2,v_2,{x}_1,v_1})$ (a multiplication of the non-positive matrix $-L^+$ and the decomposition function $\psi_d(\md{{x}_2,v_2,{x}_1,v_1})$ \yo{that} is decreasing \yo{in} \md{$(x_1,v_1)$} \md{and hence, \yo{in} $\xi_1$} by construction) and $L^-\psi_d(\md{{x}_1,v_1,{x}_2,v_2})$ (a multiplication of the non-negative matrix $L^-$ and the decomposition function $\psi_d(\md{{x}_1,v_1,{x}_2,v_2})$ \yo{that} is itself increasing \yo{in} \md{$(x_1,v_1)$ and hence, \yo{in} $\xi_1$} by construction). \yo{A s}imilar reasoning shows that \md{${f}_{\nu d}$} is decreasing in \md{$\xi_2$}. Finally, \md{${f}_{\nu d}(\xi,\xi)=\phi_d({x},w,{x},w)-L^+\psi_d({x},v,{x},x)+L^-\psi_d({x},v,{x},v)=\phi(x,w)-L\psi(x,v)=f_{\nu}(x,w,v)=f_{\nu}(\xi)$}. 
\end{proof}
We conclude this subsection by combining the results in Lemmas \ref{lem:linear_bounding} and \ref{lem:nonlinear_bounding}, as well as Proposition \ref{cor:embedding} that 
\yo{yield} the following theorem on correctness of the proposed observer.
\begin{thm}[\mk{Correct Interval Framer}]\label{lem:correctness}
Consider the nonlinear plant $\mathcal{G}$ in \eqref{eq:system} and suppose Assumptions \md{\ref{ass:initial_interval}--\ref{ass:known_input_output}} hold. Then, the dynamical system $\hat{\mathcal{G}}$ in \eqref{eq:observer} constructs a correct interval \mk{framer} for the nonlinear plant $\mathcal{G}$, \md{i.e.,} $\forall t \in \mathbb{T}, \md{\forall w_t \in \mathcal{W}, \forall v_t \in \mathcal{V}}, \underline{x}_t \leq x_t \leq \overline{x}_t$, where $x_t$ and \md{$[\underline{x}_t^\top \overline{x}_t^\top]^\top$} are the state vectors in $\mathcal{G}$ and $\hat{\mathcal{G}}$ at time $t \in \mo{\mathbb{T}}$, respectively. 
\end{thm}
\begin{proof}
It is straightforward to show that the summation of decomposition functions of constituent systems is \yo{also} a decomposition function of the summation of the constituent systems. 
\yong{Combining this with} Lemmas \ref{lem:linear_bounding} and \ref{lem:nonlinear_bounding} implies that \md{$f_d(\xi_1,\xi_2) \triangleq f_{\ell d}(\xi_1,\xi_2)+f_{\nu d}(\xi_1,\xi_2)+Ly$, with $\xi \triangleq [x^\top \ w^\top \ v^\top]^\top$} is a decomposition function for the system in \eqref{eq:eqiv_sys}, and equivalently, for System \eqref{eq:system}. Consequently, the $2n$-dimensional system \md{$\begin{bmatrix} ({\underline{x}}_t^+)^\top & (\overline{x}_t^+)^\top \end{bmatrix}^\top=\begin{bmatrix}  {f}^\top_d(\underline{\xi}_t,\overline{\xi}_t) & {f}^\top_d(\overline{\xi}_t,\underline{\xi}_t)\end{bmatrix}^\top$} {with initial condition $\begin{bmatrix} \underline{x}_0^\top & \overline{x}_0^\top\end{bmatrix}^\top$}, is an embedding system for \eqref{eq:system} (cf. Definition \ref{def:embedding}). So, $\underline{x}_t \leq x_t \leq \overline{x}_t$, by Proposition \ref{cor:embedding}.
\end{proof}
\subsection{\md{ISS and $\mathcal{H}_{\infty}$-Optimal Observer Design}}
\md{In addition to the correctness property, it is important to guarantee the} 
stability of the proposed 
\mk{framer}, 
\md{i.e, we aim to design the observer gain $L$ to ensure \yo{input-to-state stability (ISS) of}} 
the observer error, $\varepsilon_t \triangleq \overline{x}_t-\underline{x}_t$ 
(cf. Definitions \ref{defn:error} and \ref{defn:stability}). \yo{Before introducing our observer design, we first find} 
some upper bounds for \yong{the interval widths of the JSS functions in terms of the interval widths of their domains \md{through the following lemma}}. 
\begin{lem}[JSS Function \yong{Interval Width} Bounding] \label{lem:func_increment}
Let $f:\mathcal{\md{Z\triangleq \mathcal{X} \times \mathcal{W}}} \subset \mathbb{R}^\md{n_z} \to \mathbb{R}^p$ \mk{be a mapping that satisfies the assumptions in Proposition \ref{prop:JSS_decomp} and hence, can be decomposed in the form of \eqref{eq:JSS_decomp}}.
 Let $\mu_d \triangleq [\mu_{d,1}\dots \mu_{d,p}]^\top : \md{\mathcal{Z} \times \mathcal{Z}} \to \mathbb{R}^p$ be the tight decomposition function for the JSS mapping $\mu(\cdot)$, given in Proposition \ref{prop:tight_decomp}. \mk{Then,} for any \md{interval domain $\underline{z}\le z\triangleq [x^\top \ w^\top]^\top\le\overline{z}$ in 
$ \mathcal{Z}$}, 
the following inequality 
holds:
\md{
\begin{align}\label{eq:increment_bounding}
\Delta^ \mu_{d} \leq \overline{F}^{\mu} \varepsilon,
\end{align} 
\yo{where $\overline{F}^{\mu} \hspace{-.1cm} \triangleq [\overline{F}^{\mu}_x \ \overline{F}^{\mu}_w] \triangleq \hspace{-.1cm} (\overline{J}^\mu)^++\hspace{-.1cm} (\underline{J}^\mu)^-$,} 
with 
$\varepsilon \triangleq \overline{z}-\underline{z}$, $\mathcal{IZ} \triangleq [\underline{z}, \overline{z}]$, $\overline{J}^\mu=\overline{J}^f-H$, $\underline{J}^\mu=\underline{J}^f-H$, $\overline{F}^\mu_x \in \mathbb{R}^{n \times n}$ and $\overline{F}^\mu_w \in \mathbb{R}^{n \times n_w}$.} 
\end{lem}
\vspace{-.3cm}
\md{\begin{proof}
The proof follows the lines of the proof of \cite[Lemma 3]{moh2022intervalACC}, with the slight difference that here, \yo{the domain is} 
the augmentation of the state $x$ and the noise $w$. 
\end{proof}}
Now, \mk{equipped} with the tools in Lemma \ref{lem:func_increment}, we derive sufficient LMIs to synthesize the \yong{stabilizing} observer gain $L$ for both DT and CT systems through the following theorem.
\begin{thm}[\md{ISS and $\mathcal{H}_{\infty}$-Optimal Observer Design}]\label{thm:stability}
Consider the nonlinear plant $\mathcal{G}$ in \eqref{eq:system} and suppose Assumptions \md{\ref{ass:mixed_monotonicity} and \ref{ass:known_input_output}} 
hold. Then, the proposed correct interval \mk{framer} $\hat{\mathcal{G}}$ \md{in \eqref{eq:observer}} is \md{ISS}, \mk{and hence, is an interval observer} in the sense of Definition \ref{defn:stability}, \md{and also is $\mathcal{H}_{\infty}$-optimal (cf. Definition \ref{defn:H_inf})}, if there exist matrices $\mathbb{R}^{n \times n} \ni P \succ \mathbf{0}_{n,n}, \md{G \in \mathbb{R}^{n \times l}_+}$ and \md{$\gamma \in \mathbb{R}_{++}$} \md{that solve the following problem:}
\md{
\begin{align} \label{eq:H_inf_SDP}
(\gamma^*,P^*,G^*) \in \argmin\limits_{\{\gamma,P,G\}} \gamma \ \text{s.t.} \ \Gamma \prec   0, \ \text{where} 
\end{align}
}
\begin{enumerate}[(i)]
\item 
if $\mathcal{G}$ is a CT system, \md{then} 
\md{
\begin{align}\label{eq:CT_stability}
\hspace{-0.3cm}\Gamma \hspace{-.1cm} \triangleq \hspace{-.1cm} {\small\begin{bmatrix} \Omega & \Lambda & I \\ \Lambda^\top & -\gamma I & 0 \\ I & 0 & -\gamma I \end{bmatrix}}\hspace{-.05cm},{-G C} \hspace{-.075cm} \in \hspace{-.075cm}  \mathbb{M}_n, P \hspace{-.075cm}\in \hspace{-.075cm} \mathbb{D}_n, GD \hspace{-.075cm} \geq \hspace{-.075cm} 0, 
\end{align}
}\vspace{-0.2cm}

\noindent\md{with 
{$\Omega \triangleq ((A^m)+\overline{F}^{\phi}_x)^\top P+P (A^m+\overline{F}^{\phi}_x)+(-C+\overline{F}^{\psi}_x)^\top G^\top +G (-C  + \overline{F}^{\psi}_x) $}, and} 
\item if $\mathcal{G}$ is a DT system, \md{then} 
\md{
\begin{align}\label{eq:DT_stability}
\hspace{-.3cm} \Gamma \hspace{-.1cm} \triangleq \hspace{-.1cm} -\hspace{-.1cm} {\small\begin{bmatrix} P & \Omega & \Lambda & 0 \\ \Omega^\top & P & 0 & I \\ \Lambda^\top & 0 & \gamma I & 0 \\ 0 & I & 0 & \gamma I \end{bmatrix}}\hspace{-.05cm}, 
G C\hspace{-.075cm} \geq\hspace{-.075cm} 0, -P\hspace{-.075cm}  \in \hspace{-.075cm} {\mathbb{M}_n}, GD \hspace{-.075cm} \geq \hspace{-.075cm} 0,
\end{align} 
with 
{$\Omega \triangleq P (|A|+\overline{F}^{\phi}_x)+G (C  + \overline{F}^{\psi}_x)$}. 
}
\end{enumerate}
Furthermore, in both cases, \md{$\Lambda \triangleq P[\overline{F}^\phi_w+|B| \ 0]+G[0 \ \overline{F}^\psi_v+D]$ and $\overline{F}^{\phi}_x,\overline{F}^{\psi}_x,\overline{F}^{\phi}_w,\overline{F}^{\psi}_v$} 
are computed by applying Lemma \ref{lem:func_increment} on the JSS functions $\phi$ and $\psi$, respectively. Finally, the corresponding \md{optimal} stabilizing observer gain \md{$L^*$} can be obtained as \md{$L^*=(P^*)^{-1}G^*.$}
\end{thm}
\begin{proof}
Starting from \eqref{eq:observer}, we first derive the \mk{framer} error ($\varepsilon_t \triangleq \overline{x}_t- \underline{x}_t$) dynamics. Then, we show that the provided conditions in \eqref{eq:CT_stability} and \eqref{eq:DT_stability} are sufficient for stability of the error system in the CT and DT cases, respectively. To do so, define \md{$\Delta^ \mu_d  \triangleq \mu_d(\overline{x},\overline{s},\underline{x},\underline{s})-\mu_d(\overline{x},\overline{s},\underline{x},\underline{s})$ \yo{and} $\Delta s \triangleq \overline{s}-\underline{s}, \forall \mu \in \{\phi,\psi \}, s \in \{w,v\}$}. 

Now, considering the CT case, from \eqref{eq:observer} and \eqref{eq:T_CT}, we obtain  \yong{the observer error dynamics:}
\md{
{\small
\begin{align}
\nonumber  \dot{\varepsilon}_t\hspace{-.1cm}&=\hspace{-.1cm}(\hspace{-.05cm}(A\hspace{-.1cm}-\hspace{-.1cm}LC)^\text{d}\hspace{-.1cm}+\hspace{-.1cm}|(A\hspace{-.1cm}-\hspace{-.1cm}LC)^\text{nd}|)\varepsilon_t\hspace{-.1cm}+\hspace{-.1cm}\Delta^ \phi_d\hspace{-.1cm}+\hspace{-.1cm}|L||\Delta^ \psi_d| \hspace{-.1cm}+\hspace{-.1cm}|B|\Delta w \hspace{-.1cm}+ \hspace{-.1cm} |LD|\Delta v \\
&\hspace{-.cm}\leq\hspace{-.1cm}(A^\text{m}\hspace{-.1cm}+\hspace{-.1cm}(-LC)^\text{m}\hspace{-.1cm}+\hspace{-.1cm}\overline{F}^{\phi}_x\hspace{-.1cm}+\hspace{-.1cm}|L|\overline{F}^{\psi}_x)\varepsilon_t\hspace{-.1cm}+\hspace{-.1cm}\delta^{w,v}(L),\label{eq:error_dynamics_2}
\end{align} 
}
}\vspace{-0.35cm}

\noindent where \md{$\delta^{w,v}(L) \triangleq (\overline{F}^\phi_w+|B|)\Delta w + (|L|\overline{F}^\psi_v+|LD|)\Delta v $ \yo{and} $\yo{\overline{F}^{\mu}_s,} \ \forall \mu \in \{\phi,\psi \},s \in \{w,v\}$} is given in \eqref{eq:increment_bounding}. \yo{The} inequality holds by Lemma \ref{lem:func_increment}, Proposition \ref{prop:bounding}, and the facts that $\forall M,N \in \mathbb{R}^{n \times n}$, $(M+N)^\text{d}=M^\text{d}+N^\text{d},(M+N)^\text{nd}=M^\text{nd}+N^\text{nd}$, $|M+N| \leq |M|+|N|$ by triangle inequality and the fact that $\varepsilon_t \geq 0$ by the correctness \yong{property} (Lemma \ref{lem:correctness}). Now, note that by the \emph{Comparison Lemma} \cite[Lemma 3.4]{khalil2002nonlinear} and positivity of the system in 
 \eqref{eq:error_dynamics_2}, stability of the system in \eqref{eq:error_dynamics_2} implies 
\yong{stability} for the actual error system. To show the former, 
we \yong{require the following:} 
\md{$G$ and $P$} 
\yong{are} \md{non-negative} and diagonal matrices, respectively. This forces \md{$P$} and its inverse to be diagonal matrices with strictly positive diagonal elements, and since \md{$G$} is forced to be non-\md{negative}, \md{$L=P^{-1}G$} must be non-negative, and hence $|L|=L$. 
\md{Moreover, $-GC$} is Metzler, {\md{which} results in \md{$-LC=-P^{-1}G C$} being Metzler, since} 
\mk{it} 
is a product of 
a diagonal and positive matrix \md{$P^{-1}$} and 
a Metzler matrix \md{$-G C$}. 
Thus, $(-LC)^\text{m}=-LC$. \md{Further, since $GD$ is non-negative, then $LD=P^{-1}GD$ is a product of two non-negative matrices $P^{-1}$ and $GD$ and so, $|LD|=LD$.}
\md{Hence}, the system in \eqref{eq:error_dynamics_2} \yo{becomes} the linear comparison system

\phantom{a}\vspace{-.3cm}
{\small
\begin{align}\label{eq:comparison}
 \dot{\varepsilon}_t \hspace{-.1cm} \le \hspace{-.1cm}  (A^\text{m}\hspace{-.1cm} -\hspace{-.1cm} LC\hspace{-.1cm} +\hspace{-.1cm} \overline{F}_{\phi}\hspace{-.1cm} +\hspace{-.1cm} L\overline{F}_{\psi})\varepsilon_t\md{+\hspace{-.1cm} L(\overline{F}^\psi_v\hspace{-.1cm} +\hspace{-.1cm} D)\Delta v\hspace{-.1cm} +\hspace{-.1cm} (\overline{F}^\phi_w\hspace{-.1cm} +\hspace{-.1cm} |B|)\Delta w},
 \end{align} 
 }\vspace{-.3cm}
 
 \noindent\md{where by \cite[Sec. 9.2.2]{duan2013lmis}, solving the SDP in \eqref{eq:H_inf_SDP},\eqref{eq:CT_stability} results in the optimal observer gain $L^*=(P^*)^{-1}G^*$, in the $\mathcal{H}_{\infty}$ sense, i.e., \eqref{eq:H_inf_Def} holds with $\gamma^*$. This implies that the above linear comparison system \eqref{eq:comparison} satisfies the following asymptotic gain (AG) property \cite{sontag1996new}:
\begin{align}\label{eq:AG}
\limsup_{t \to \infty} \|\varepsilon_t\|_{\infty} \hspace{-.1cm} \leq \hspace{-.1cm} \rho(\| \yo{\tilde{\Delta}}\|_{\infty}), \ \forall \varepsilon_0,\forall \yo{\tilde{\Delta}\hspace{-.1cm}\in\hspace{-.1cm}[\Delta {w}^\top \ \Delta {v}^\top ]^\top,} 
\end{align}
where $\tilde{\Delta}$ is any realization of the augmented noise \yo{interval} width and $\rho$ is any class $\mathcal{K}_{\infty}$ function that is lower bounded by $\gamma^*\tilde{\Delta}$. On the other hand, by setting $\Delta=0$, the LMIs in \eqref{eq:CT_stability} 
\yo{reduce} to their noiseless counterparts in \yo{\cite[Eq. (19)]{moh2022intervalACC}}. Hence, by \yo{\cite[Theorem 2]{moh2022intervalACC}}, the comparison system \eqref{eq:comparison} is 0-stable (0-GAS), which in addition to the AG property \eqref{eq:AG} is equivalent to the ISS property for \eqref{eq:comparison} by \cite[Theorem 1-e]{sontag1996new}. Hence, the designed CT observer is also ISS. 
} 

For the DT case, from \eqref{eq:observer} and \eqref{eq:T_DT} and by \yo{a} similar reasoning to the CT case, we obtain 
\md{ 
\begin{align}\label{eq:error_dynamics_2_DT}
\begin{array}{rl}
{\varepsilon}_{t+1}&=|A\hspace{-.1cm}-\hspace{-.1cm}LC|\varepsilon_t\hspace{-.1cm}+\hspace{-.1cm}\Delta^ \phi_d\hspace{-.1cm}+\hspace{-.1cm}|L||\Delta^ \psi_d|\hspace{-.1cm}+\hspace{-.1cm}|B|\Delta w \hspace{-.1cm}+ \hspace{-.1cm} |LD|\Delta v \\
&\leq(|A|+|LC|+\overline{F}^{\phi}_x+|L|\overline{F}^{\psi}_x)\varepsilon_t+\delta^{w,v}(L).
\end{array} 
\end{align} 
}\vspace{-0.2cm}

\noindent  In addition, we enforce \md{$-P$} to be Metzler, 
  as well as \md{$G$ and $G C$ to be non-negative}. 
  Consequently, {since \md{$P$} is positive definite,} 
  \md{$P$} becomes a non-singular M-matrix\footnote{An M-matrix is a square matrix \yong{whose negation is Metzler and whose eigenvalues have nonnegative real parts}.}, and hence is inverse-positive \cite[Theorem 1]{plemmons1977m}, i.e., \md{$P^{-1} \geq 0$}. Therefore, \md{$L=P^{-1}G \geq 0$} and \md{$LC=P^{-1}(G C) \geq 0$}, because they are matrix products of non-negative matrices, \md{$P^{-1},G$ and $P^{-1},G C$}, respectively. \md{Finally, by a similar argument as in the CT case, $LD$ is non-negative}. Hence, $|L|=L,|LC|=LC,\md{|LD|=LD}$, and so, the system in \eqref{eq:error_dynamics_2_DT} \yo{becomes} 
  \begin{align}\yo{\small\begin{array}{rl}{\varepsilon}_{t+1}\le &(|A|+LC+\overline{F}^{\phi}_x+L\overline{F}^{\psi}_x)\varepsilon_t \\
  &\md{+L(\overline{F}^\psi_v+D)\Delta v+(\overline{F}^\phi_w+|B|)\Delta w},\end{array}}\end{align} 
  \md{for} which \yo{the solution to the} \md{
  the SDP in \eqref{eq:H_inf_SDP},\eqref{eq:DT_stability} \yo{provides} 
  the $\mathcal{H}_{\infty}$-optimal observer gain $L^*=(P^*)^{-1}G^*$, by \cite[Sec. 9.2.3]{duan2013lmis}. Furthermore, a similar argument as in the CT case implies that the DT observer is also ISS.} 
\end{proof}

\vspace{-0.1cm}
\mk{Finally, note that 
\md{\yo{if} the LMIs in \eqref{eq:CT_stability} 
or \eqref{eq:DT_stability} are infeasible}, 
a coordinate transformation \md{can be applied in a straightforward manner, similar to \cite[Section V]{tahir2021synthesis}} (omitted due to space limitations; \md{also} cf. \cite{mazenc2021when} and references therein), \md{which} may also be helpful for making the LMIs in Theorem \ref{thm:stability} feasible, as observed in Section \ref{sec:CT_exm}.}

\vspace{-0.1cm}
\section{Illustrative Examples}
The effectiveness of our observer design is illustrated for CT and DT systems \mk{(using \yo{SeDuMi \cite{sturm1999using}}  \md{to solve the LMIs}).}

\vspace{-0.1cm}
\subsection{CT System Example}\label{sec:CT_exm}
Consider 
the CT system
in \mk{\cite[Section IV, Eq. (30)]{dinh2014interval}}:
\md{
{\small
\begin{gather*}
\dot{x}_{1} = x_{2}+w_1, \quad \dot{x}_{2}=b_1x_3-a_1\sin(x_1)-a_2x_2+w_2, \\ 
 \dot{x}_3\hspace{-.1cm}=\hspace{-.1cm}-a_3(a_2x_1\hspace{-.1cm}+\hspace{-.1cm}x_2)\hspace{-.1cm}+\hspace{-.1cm}\frac{a_1}{b_1}(a_4\sin(x_1)\hspace{-.1cm}+\hspace{-.1cm}\cos(x_1)x_2)\hspace{-.1cm}-\hspace{-.1cm}a_4x_3\hspace{-.1cm}+\hspace{-.1cm}w_3,
\end{gather*}}\vspace{-0.35cm}

\noindent with 
output $y=x_1$, $a_1=35.63, b_1=15, a_2=0.25,a_3=36,a_4=200, \mathcal{X}_0 = [19.5, 9] \times [9, 11] \times [0.5,1.5], \md{\mathcal{W}=[-0.1,0.1]^3}$. 
}
Without a coordinate transformation, the LMIs in \eqref{eq:CT_stability} \md{as well as the approach in \cite{dinh2014interval}} were infeasible. \yo{However,} with a coordinate transformation $z=Tx$ with $T={\footnotesize{\begin{bmatrix}20 & 0.1 & 0.1\\ 0 & 0.01 & 0.06\\0 &-10 & -0.4\end{bmatrix}}}$ \yo{(similar to \cite[Section V]{tahir2021synthesis})} and adding and subtracting $5y$ to the dynamics of $\dot{x}_1$, 
 the state framers \mk{returned by our approach, 
 $\underline{x},\overline{x}$ are tighter than the ones obtained by the interval observer in \cite{dinh2014interval}, $\underline{x}^{DMN},\overline{x}^{DMN}$ (primarily because of outer-approximations of the initial framers $\mathcal{X}_0$ due to different coordinate transformations), \md{\yo{as shown in 
Figure \ref{fig:figure1} ($x_1,x_2$ omitted for brevity)}}. Further, 
the framer error $\varepsilon_t= \overline{x}_t-\mkr{\underline{x}}_t$ \md{is smaller for our approach \yo{when} compared to to the one in 
\cite{dinh2014interval} and} is observed to 
tend to \md{steady state} asymptotically.} 

 \subsection{DT System Example}\label{sec:DT_exm}
Consider a \md{noisy} variant of \yo{the} 
H\'enon chaos system \cite{Observer_discrete}:
\vspace{-.3cm}
\md{
\begin{align}
\label{eq:exampletwo}   
x_{t+1} =  Ax_t + r[1 -x_{t,1}^2 ]+Bw_t, \quad y_t =  x_{t,1}+v_t,
\end{align}
}\vspace{-0.2cm}

\noindent where 
$A =
\small{\begin{bmatrix}
0 & 1\\
0.3 & 0 
\end{bmatrix}}$, \md{$B=I$}, $r =\small{\begin{bmatrix}0.05 \\ 0\end{bmatrix}}$, $\mathcal{X}_0 = [-2, 2] \times [-1, 1]$, \md{$\mathcal{W}=[-0.01,0.01]^2$ and $\mathcal{V}=[-0.1,0.1]$}. 
\md{\yo{Using the solutions to} the corresponding LMIs in \eqref{eq:DT_stability}, it can be observed \yo{from} Figure \ref{fig:figure2} that the interval estimates for $x_2$ are tighter than the ones returned by the approach in \cite{tahir2021synthesis} (similarly for $x_1$, omitted for brevity). Moreover, the depicted error plots demonstrate the convergence of \yo{the} error sequence to steady state (i.e., \yo{ISS}) \yo{and show} smaller errors for the proposed approach \yo{when} compared to the one in 
\cite{tahir2021synthesis}.} 
\begin{figure}[t!] 
\centering
{\includegraphics[width=0.48\columnwidth,trim=5mm -5mm 10mm 0mm]{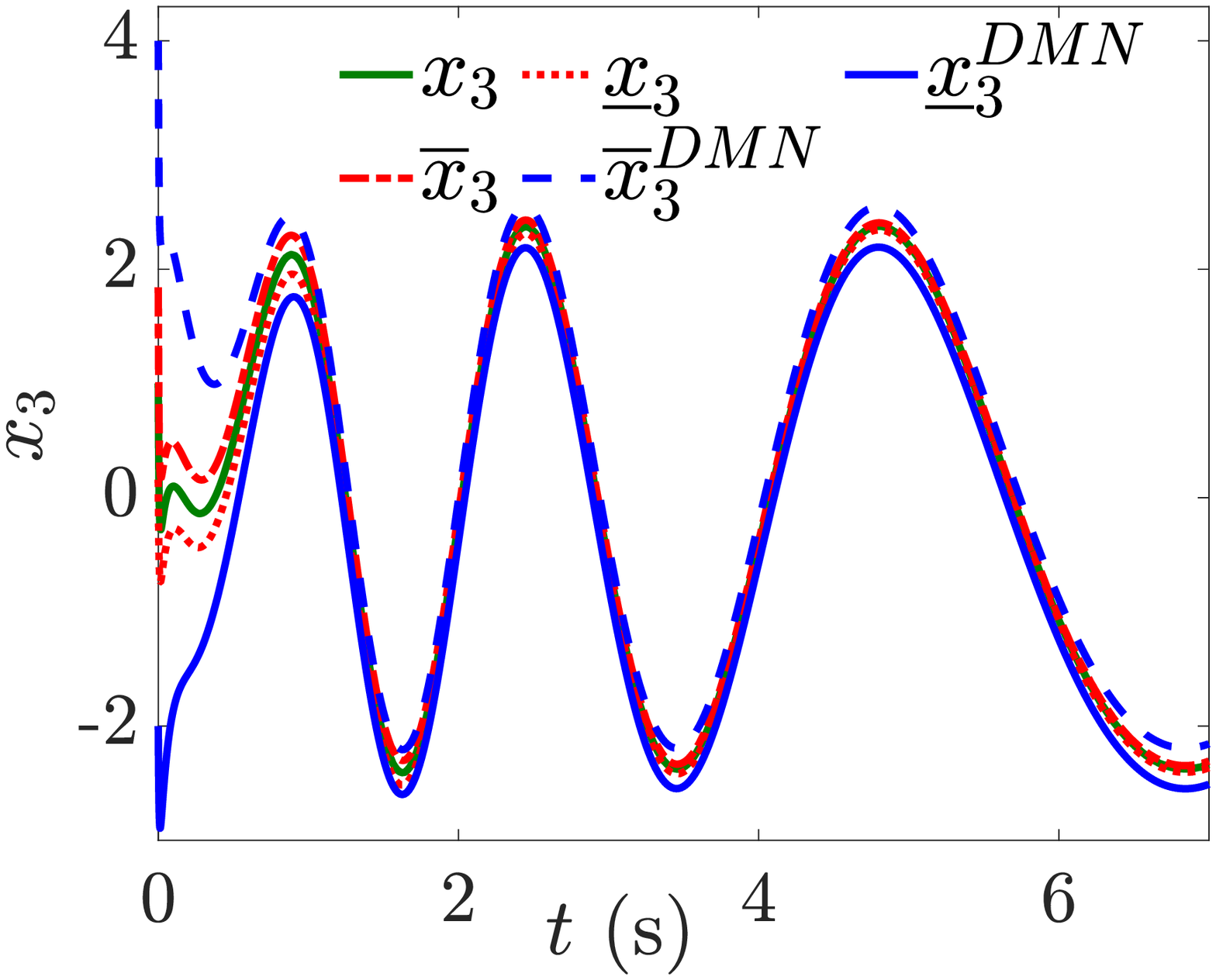}} \label{fig:sub1}
{\includegraphics[width=0.48\columnwidth,trim=5mm -5mm 10mm 0mm]{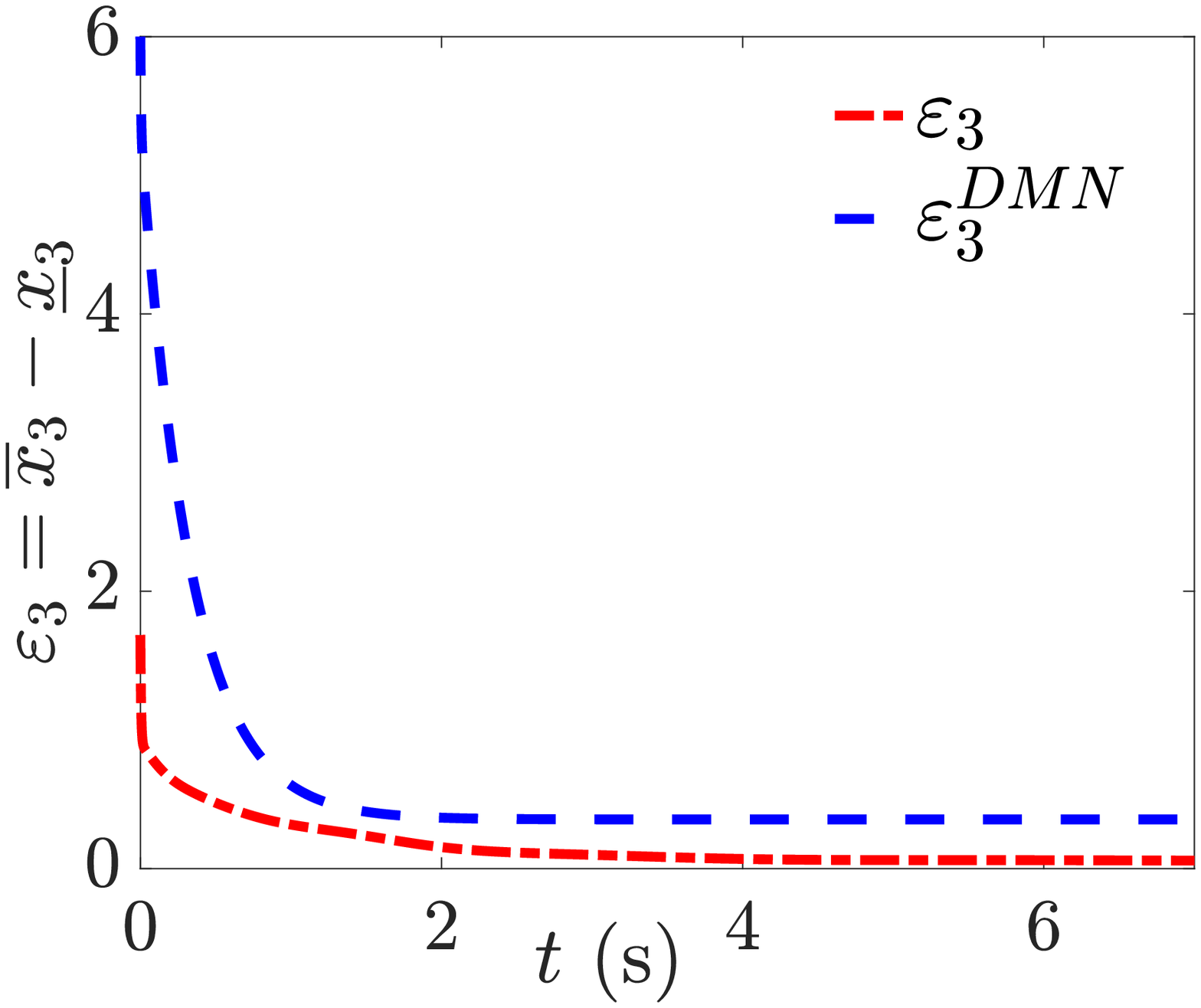}} \label{fig:sub2}
\vspace{-.3cm}
\caption{{{\small State, $x_3$, as well as its upper and lower framers and error \mk{returned by our proposed observer}, $\overline{x}_3,\underline{x}_3,\varepsilon_3$, \mk{and by the observer in \cite{dinh2014interval}, $\overline{x}^{DMN}_3,\underline{x}^{DMN}_3,\varepsilon^{DMN}_3$} for the CT System example.}}}
\label{fig:figure1}
\end{figure}

\begin{figure}[t!] 
\centering
{\includegraphics[width=0.48\columnwidth,trim=10mm 5mm 10mm 0mm]{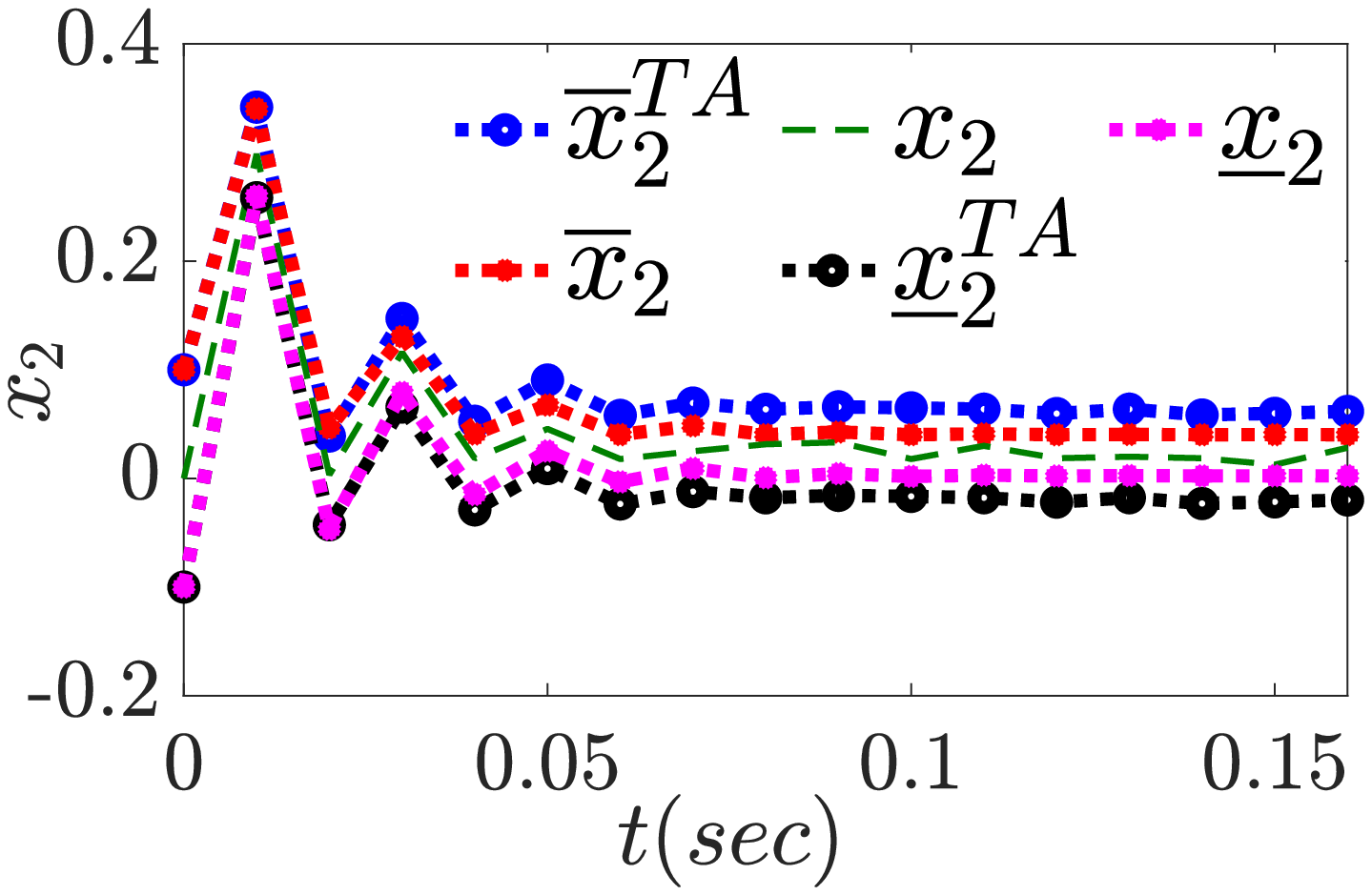}}\label{fig:sub4}
{\includegraphics[width=0.48\columnwidth,trim=7mm 5mm 20mm 0mm]{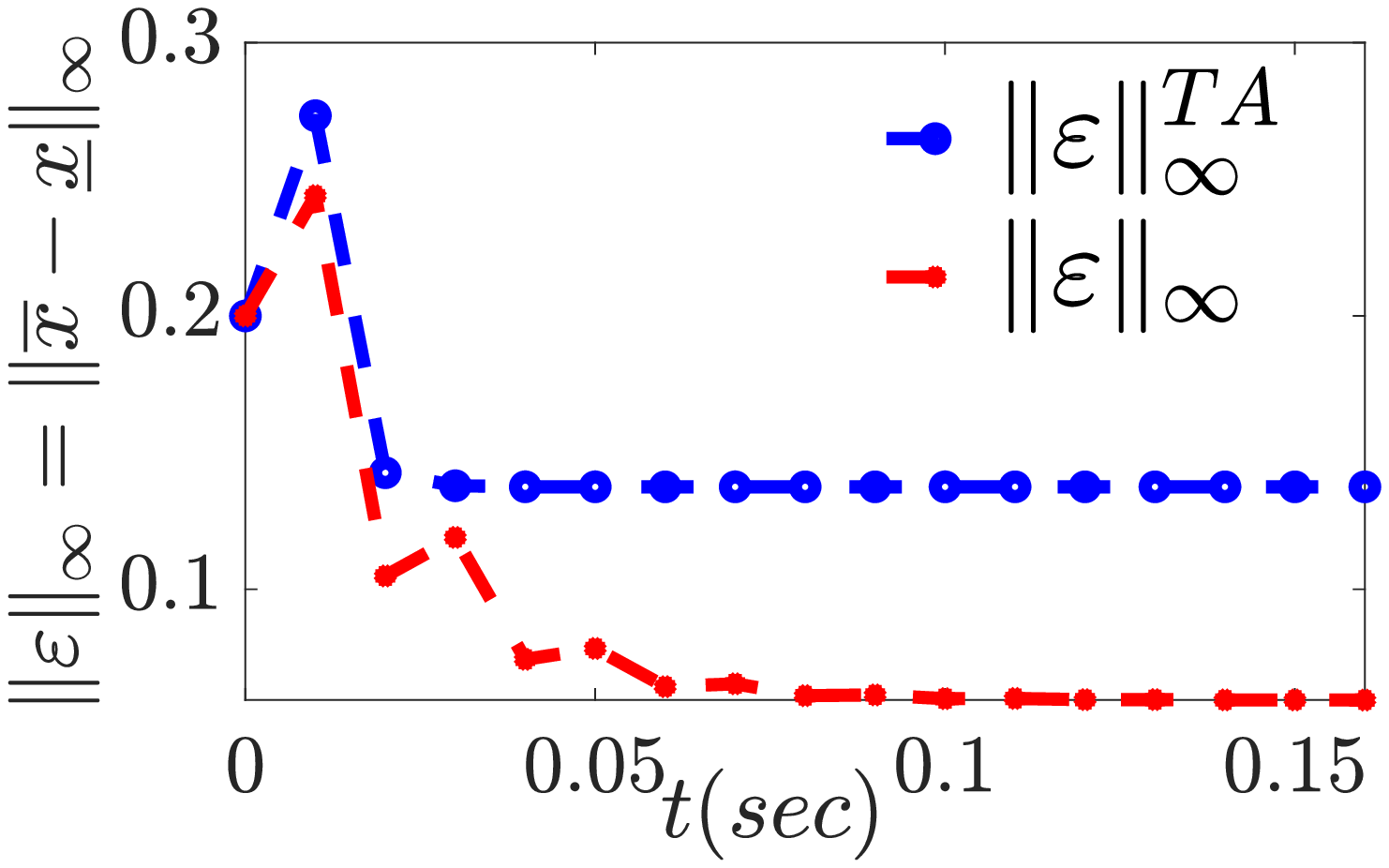}}\label{fig:sub5}
\caption{\small\mk{State, $x_2$, and its upper and lower framers, \mk{returned by our proposed observer}, $\overline{x}_2,\underline{x}_2$, \mk{and by the observer in \cite{tahir2021synthesis}, $\overline{x}^{TA}_2,\underline{x}^{TA}_2$} (left) and norm of framer error (right) for the DT System example.}}
\label{fig:figure2}
\end{figure}

\section{Conclusion and \md{Future Work}} \label{sec:conclusion}
\md{A novel \yo{unified} approach to \yo{synthesize} interval-valued observers for bounded-error locally Lipschitz nonlinear continuous-time (CT) and discrete-time (DT) systems with nonlinear noisy observations was presented. The proposed observer was shown to be correct by construction using mixed-monotone decompositions, \yo{i.e.,} 
the true state  trajectory of the system \yo{is guaranteed} 
to be framed by the states of the observer \yo{without the need for additional constraints or assumptions}. Moreover, \yo{we provide semi-definite programs for both CT and DT cases to find} input-to-state stabilizing observer gains 
\yo{that are} 
proven to be optimal in the sense of $\mathcal{H}_{\infty}$. Finally, simulation results demonstrated \yo{the} better performance of the \yo{proposed interval observers when} compared to some benchmark CT and DT interval observers. Designing hybrid interval observers and \yo{considering unbounded} unknown inputs 
will be considered \yo{in} our future work.}
\bibliographystyle{unsrturl}
\bibliography{biblio}

\end{document}